\documentclass{amsart}
\usepackage{graphicx}
\usepackage{epsfig}
\usepackage{geometry} 
\usepackage{amsmath}
\usepackage{amssymb}
\usepackage{latexsym}

\newtheorem{theorem}{Theorem}
\newtheorem{lemma}[theorem]{Lemma}
\newtheorem{corollary}[theorem]{Corollary}

\newcommand{\RR}{{\mathbb R}}

\newcommand{\Cc}{\mathcal C}
\newcommand{\Ac}{\mathcal A}
\newcommand{\Hc}{\mathcal H}
\newcommand{\Pc}{\mathcal P}

\newcommand{\Nc}{\mathcal N}
\begin{document}

\title[Bureaucratic  sets]{`Bureaucratic' set systems, and their role in phylogenetics}
\author{David Bryant and Mike Steel}

\thanks{}

\address{David Bryant: Department of Mathematics and Satistics, University of Otago, Dunedin, New Zealand.\\
Mike Steel: Department of Mathematics and  Statistics, University of Canterbury, Christchurch, New Zealand}

\email{david.bryant@otago.ac.nz, mike.steel@canterbury.ac.nz}


\begin{abstract}
We say that a collection $\Cc$ of subsets of $X$ is  {\em bureaucratic} if every maximal hierarchy on $X$ contained in $\Cc$ is also maximum. We characterise bureaucratic set systems and show how they arise in phylogenetics. This framework has several useful algorithmic consequences: we generalize some earlier results and  derive a polynomial-time algorithm for a parsimony problem arising in phylogenetic networks.
\end{abstract}

\maketitle

\section{Bureaucratic sets and their characterization}

We first recall some standard phylogenetic terminology (for more details, the reader can consult \cite{sem}).
Recall that a {\em hierarchy} $\Hc$ on a finite set $X$  is a collection of sets with the property that the intersection of any two sets is either empty or equal to one of the two sets; we also assume that $X \in \Hc$.

\begin{figure}[ht]
\begin{center}
\resizebox{11cm}{!}{
\includegraphics{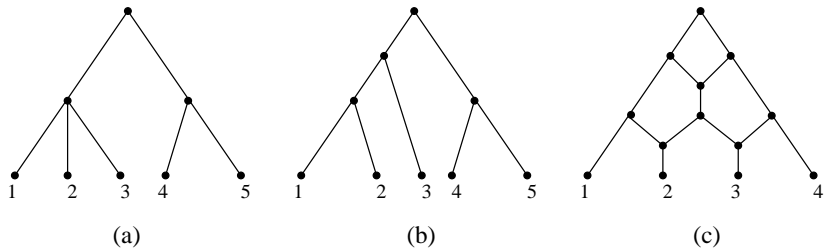}
}
\caption{(a): A rooted tree $T$ with leaf set $X=\{1,2,3,4,5\}$, and with cluster set $c(T)$ being equal to  the hierarchy $\Hc$ consisting of the sets $\{1,2,3\}, \{4,5\}$ and the trivial clusters. (b): A binary tree $T$ with a cluster set consisting  of $\Hc \cup \{\{1,2\}\}$. (c): A binary and planar phylogenetic network $\Nc$ over $X=\{1,2,3,4\}$
with a soft-wired cluster set $sw(\Nc)$ consisting of $\{1,2\}, \{2,3\}, \{3,4\}, \{1,2,3\}, \{2,3,4\}$ and the trivial clusters.}
\label{fig1}
\end{center}
\end{figure}

\noindent A hierarchy is {\em maximum} if $|\Hc| = 2|X|-1$, which is the largest possible cardinality. In this case $\Hc$ corresponds to the set of clusters $c(T)$ of some rooted binary tree $T$ with leaf set $X$ (a {\em cluster} of $T$ is the set of leaves that are separated from the root of the tree by any vertex).   A maximum hierarchy necessarily contains $\{x\}$ for each $x \in X$, as well as  $X$ itself;  we will refer to these $|X|+1$ sets as the {\em trivial clusters} of $X$. More generally, any hierarchy containing all the trivial clusters corresponds to the clusters $c(T)$ of a rooted tree $T$ with leaf set $X$ (examples of these concepts are illustrated in Fig.~\ref{fig1}(a),(b)).  Note that a hierarchy $\Hc$ is maximum if and only if  (i) $\Hc$ contains all the trivial clusters, and (ii) each set $C \in \Hc$ of size greater than 1 can be written
as a disjoint union $C= A\sqcup B$, for two (disjoint) sets $A,B \in \Hc$.   

We now introduce a new notion.

\bigskip
\noindent {\bf Definition:}  We say that a  collection $\Cc$ of subsets of a finite set $X$ is a {\em bureaucracy} if (i) $\Cc \neq \emptyset$ and $\emptyset \not\in \Cc$, and (ii) every hierarchy $\Hc \subseteq \Cc$ can be extended to a maximum hierarchy $\Hc'$ such that $\Hc \subseteq \Hc' \subseteq \Cc$. In this case, we say that the collection is {\em bureaucratic}. 

\bigskip

Simple examples of bureaucracies  include two extreme cases:  the set of clusters of a binary tree, and  the set $\Pc(X)$ of all non-empty subsets of $X$.    Notice that $\{\{a\}, \{b\}, \{c\}, \{a,b\}, \{a,b,c\}\}$ and $\{\{a\}, \{b\}, \{c\}, \{b,c\}, \{a,b,c\}\}$ are both bureaucratic subsets of $X=\{a,b,c\}$ but their intersection, 
$\{\{a\}, \{b\}, \{c\}, \{a,b,c\}\}$, is not. In particular,  for an arbitrary subset $Y$ of $X$ (e.g.  $Y=\{\{a\}, \{b\}, \{c\}, \{a,b,c\}\}$), there may not be a unique minimal
bureaucratic subset of $X$ containing $Y$.

In the next section we describe a more extensive list of examples, but first we describe some properties and provide a characterization of bureaucracies.  
In the following lemma, given two sets $A$ and  $B$ from $\Cc$ we say that $B$ {\em covers} $A$  if $A \subsetneq B$ and there is no set $C \in \Cc$ with $A \subsetneq C \subsetneq B$.
\begin{lemma}
\label{mainlem}
\mbox{}
If  $\Cc$ is bureaucratic then:
\begin{itemize}
\item[{\rm (i)}] For any pair $A,B \in \Cc$, if $B$ covers $A$ then $B - A \in \Cc$.
\item[{\rm (ii)}] For  any  $C \in \Cc$ with $|C|>1$,  we can write $C = A \sqcup B$ for  (disjoint) sets $A,B \in \Cc$.
\end{itemize}
\end{lemma}
\begin{proof}  
For Part (i),  suppose that $A,B \in \Cc$ and that $B$ covers $A$.  Let $\Hc= \{A, B\}$.  Then $\Hc$ is a hierarchy that is contained within $\Cc$ and so there exists a maximum hierarchy $\Hc' \subseteq \Cc$ that contains $\Hc$. 
Note that $A$ must be a maximal sub-cluster of $B$ in $\Hc'$ (as otherwise $B$ does not cover $A$) which requires that $B-A$ is a cluster of $\Hc'$ and thereby an element of $\Cc$.  

For Part (ii), observe that the set $\Hc=\{C\}$ is a hierarchy, and  the assumption that $\Cc$ is bureaucratic ensures  the existence of a maximum hierarchy $\Hc' \subseteq \Cc$ containing $\Hc$, and so $\Hc'$ contains the required sets $A,B$. 
\end{proof}

Note that the conditions described in Parts (i) and (ii) of Lemma~\ref{mainlem}, while they are necessary for $\Cc$ to be a bureaucracy, are not sufficient. For example, let
$X=\{1,2,3,4,5,6\}$ and let $\Cc$ be the union of $$\{\{1,2\}, \{3,4\}, \{5,6\}, \{1,2,3\}, \{4,5,6\}, \{3,4,5\}, \{1,2,6\}, \{1,5,6\}, \{2,3,4\}\}$$ with the 
set of the seven trivial clusters. Then $\Cc$ satisfies Parts (i) and (ii) of Lemma~\ref{mainlem}, yet $\Cc$ is not bureaucratic since $\Hc = \{\{1,2\}, \{3,4\}, \{5,6\}\}$ does not extend to a maximum hierarchy on $X$ using just elements from $\Cc$.

\begin{theorem}
\label{mainthm}
\mbox{}
 A collection $\Cc$ of subsets of $X$ is bureaucratic if and only if it satisfies the following two properties:
\begin{itemize}
\item{{\rm (P1)}} $\Cc$ contains all trivial clusters of $X$.
\item{{\rm (P2)}}If $\{C_1, C_2, \ldots,C_k\} \subseteq \Cc$ are disjoint and have union $\cup_i C_i$ in $\Cc$ then there are distinct $i,j$ such that $C_i \cup C_j \in \Cc$. 
\end{itemize}
\end{theorem}

\begin{proof}
First suppose that $\Cc$ is bureaucratic.  Then $\Cc$ contains a maximal hierarchy; in particular, it contains all the trivial clusters, and so (P1) holds. For (P2), suppose that  $\Cc'$ is a collection of $k \geq 3$ disjoint  subsets of $X$, each an element of $\Cc$, and $\bigcup \Cc' \in \Cc$.  Then $\Hc =\Cc' \cup \{\bigcup \Cc'\}$ is a hierarchy. Let $\Hc'\subseteq \Cc$ be a maximal hierarchy on $X$ that contains $\Hc$ (this exists, since $\Cc$ is bureaucratic)  and let $C$ be a minimal subset of $X$ in  $\Hc'$ that contains the union of at least two elements of $\Cc'$.  Since $\Hc'$ is a binary hierarchy, and $\bigcup \Cc' \in \Hc'$, $C$ is precisely the union of exactly two  elements of $\Cc'$; since $C \in \Hc' \subseteq \Cc$, this establishes (P2).

Conversely, suppose that  a collection $\Cc$ of subsets of $X$ satisfies (P1) and (P2), and that $\Hc \subseteq \Cc$ is a maximal hierarchy which is contained within $\Cc$.  Suppose that $\Hc$ is not maximum (we will derive a contradiction). Then $\Hc$ contains a set $C$ that is the disjoint union of $k \geq 3$ maximal proper subsets $A_1, \ldots, A_k$, each belonging to $\Hc$ (and thereby $\Cc$).  Applying (P2) to  $\Cc' =\{A_1, \ldots, A_k\}$, there exists two sets, say $A_i, A_j$ for which $A_i \cup A_j \in \Cc.$  So, if we let $\Hc'=\Hc \cup \{A_i \cup A_j\}$, then we obtain a larger hierarchy containing $\Hc$ that is still contained within $\Cc$, which is a contradiction. 
This completes the proof. 
\end{proof}

\section{Examples of bureaucracies}
We have mentioned two extreme cases of bureaucracies, namely the set of clusters of a binary $X-$tree and  the full power set $\Pc(X)$. Here are some  further examples.
\begin{enumerate}

\item The set of intervals of $[n]=\{1,2,\ldots,n\}$ is a bureaucracy,  where an {\em interval} is a set $[i,j] = \{k:i\leq k\leq j\}$, $1 \leq i \leq j \leq n$.
\begin{proof}
Let $\Cc$ be the set of intervals of $[n]$. Then $\Cc$ contains the trivial clusters. Also, a disjoint collection $I_1, \ldots, I_k$, $k>2$ of intervals has union an interval if and only if every element of $[n]$ between $\min\bigcup I_j$ and $\max \bigcup I_j$ lies in  (exactly) one interval, in which case the union of any pair of consecutive intervals is an interval, so (P2) holds.  By Theorem \ref{mainthm}, $\Cc$ is bureaucratic. 
\end{proof}

Similarly, if we order the elements of $X$ in any fashion, we can define the set of {\em intervals on $X$} for that ordering by this construction (associating $x_i$ with $i$), and can thus obtain a bureaucracy. 

A natural question at this point is the following: Does the extension of intervals in a 1-dimensional lattice (Example 1) to rectangles in a 2-dimensional lattice also necessarily lead to bureaucracies?  The answer is `no' because condition (P2) can be violated due to the existence of subdivisions of  integral sized rectangles into $k>2$ disjoint squares of different integral sizes, the union of any two of which must therefore fail to be a rectangle (see e.g. \cite{brooks}).

\item  Let $T$ be a rooted tree (generally not binary) with leaf set $X$ and let $\Cc$ be the set of all clusters compatible with all the clusters in $c(T)$. Then $\Cc$ is bureaucratic.
\begin{proof}
We have $\Cc = \{C \subseteq X: C \cap C' \in \{C, C', \emptyset\} \mbox{ for all }  C' \in c(T)\}.$  $\Cc$ is also  the set of clusters that occur in at least one rooted phylogenetic $X-$ tree that refines $T$, that is:
$$\Cc = \bigcup_{T': c(T) \subseteq c(T')} c(T').$$ 
Suppose that $\Hc \subseteq \Cc$ is a hierarchy on $X$.  Then $\Hc \cup c(T)$ is also a hierarchy on $X$ since every element of $\Hc$ is compatible with every element of $c(T)$.
Let $\Hc'$ be any maximal hierarchy on $X$ containing $\Hc$. Then since $c(T) \subseteq \Hc'$, we have $\Hc' \subseteq \Cc$, and so, by definition, $\Cc$ is a bureaucracy.
\end{proof}

\item
Let $\Cc$ be a collection of subsets of $X$ that includes the trivial clusters and which satisfies the condition:
\begin{equation}
\label{closure}
A, B \in \Cc \mbox{ and } A \cap B \neq \emptyset \Rightarrow A \cup B \in \Cc.
\end{equation}
Then $\Cc$ is bureaucratic if and only if  $\Cc$ satisfies the covering condition in Lemma \ref{mainlem}(i). 

Condition (\ref{closure}) is a weakening of the condition
required for a `patchwork' set system on $X$ due to Andreas Dress and Sebastian B{\"o}cker (see e.g. \cite{sem}, where the covering condition of Lemma \ref{mainlem}(i) leads to an `ample patchwork').

\begin{proof}
The `only if' part follows from Lemma  \ref{mainlem}(i).  Conversely, suppose that (\ref{closure}) holds for a set system $\Cc$ that includes all the trivial clusters of $X$ and that satisfies the covering condition of Lemma~\ref{mainlem}(i).
Suppose that $\Hc \subseteq \Cc$ is a maximal hierarchy contained within $\Cc$. We show that $\Hc$  is maximum.  Suppose that this is not the case -- we will derive a contradiction (by constructing a larger hierarchy $\Hc'$ containing $\Hc$ but still lying within $\Cc$).  The assumption that $\Hc$ is not maximum implies that there exists a set $B \in \Hc$ which is the union of three or more disjoint sets $A_1, A_2, A_3,\ldots, A_k$, where $A_i \in \Hc$ (since the rooted tree associated with $\Hc$ has a vertex of degree $k \geq 3$). 
We consider two cases:
\begin{itemize}
\item[(i)] $B$ covers none of the sets from $A_1, A_2, A_3,\ldots, A_k$. 
\item[(ii)] $B$ covers one of the sets from $A_1, A_2, A_3,\ldots, A_k$. 
\end{itemize} 
We first show that Case (i) cannot arise under Condition  (\ref{closure}). Suppose to the contrary Case (i) arises. Then  for $i=1, \dots k$ there exists a set $C_i \in \Cc$ that contains $A_i$ and which is covered by $B$. For any pair $i,j$ with $i \neq j$,  if  $(B-C_i) \cap C_j = \emptyset$ then $C_j \subseteq C_i$. On the other hand, if $(B-C_i) \cap C_j \neq \emptyset$ then, by Condition  (\ref{closure}), 
$(B -C_i) \cup C_j \in \Cc$, which means that $B = (B-C_i) \cup C_j$ (otherwise $(B-C_i) \cup C_j$ an element of $\Cc$ strictly containing $C_j$ and strictly contained by $B$) and so
$C_i \subseteq C_j$.  Thus Case (i) requires that either  $C_i \subseteq C_j$ or $C_j \subseteq C_i$, which implies (again by the assumption that $B$ covers $C_i$ and $B$ covers $C_j$)  that $C_i=C_j$.  Since this identity holds for all distinct pairs $i,j$ it follows that $C_1, C_2, \ldots, C_k$ are the same set $C$ and this set contains 
$\bigcup_{i=1}^k A_i$ (since $A_i \subset C_i$).   But then  $B= \bigcup_{i=1}^k A_i  \subseteq C$ which contradicts the assumption that $B$ covers $C_1 (=C)$.

Thus only Case  (ii) can arise. In this case, suppose that $B$ covers $A_i$.  By assumption that $\Cc$ satisfies the covering condition described in Lemma \ref{mainlem}(i), $B-A_i \in \Cc$ holds, and so we can take $\Hc'= \Hc \cup \{B-A_i\}$ which provides the required contradiction.

\end{proof}

\item
Let $G = (X,E)$ be a connected graph. Let $\Cc$ be the set of subsets $Y \subseteq X$ such that $G[Y]$ is connected (where $G[Y]$ is the subgraph formed by deleting vertices not in $Y$, together with their incident edges).
Then $\Cc$ is bureaucratic.

Observe that taking $G$ to be a linear graph recovers Example (1). 
 
\begin{proof}
First note that $\Cc$ satisfies (P1), since $G$ itself is connected, as is each vertex by itself. Now suppose that $A_1,\ldots,A_k$, $k>2$, are disjoint clusters in $\Cc$ whose union, $A$, is also in $\Cc$. As $G[A]$ is connected, at least two clusters $A_i$, $A_j$ must contain adjacent vertices, in which case $G[A_i \cup A_j]$ is connected and $A_i \cup A_j \in \Cc$. The result now follows by Theorem \ref{mainthm}.
 
 An alternative proof is to apply Example (3) and note that $\Cc$ satisfies  Condition (\ref{closure}) and the covering condition of Lemma \ref{mainlem}(i).
 \end{proof}

\end{enumerate}

\section{Algorithmic applications}

\subsection{Maximum weight hierarchies}

In general, the problem of finding the largest hierarchy contained within a set of clusters is NP-hard \cite{day}. The problem becomes trivial in a bureaucratic collection since all maximal hierarchies are maximum. Less obvious, however, is the fact that finding a hierarchy with maximum {\em weight} can also be solved in polynomial time.

\begin{theorem}
Let $\Cc$ be a bureaucratic collection of clusters on $X$ and let $w:\Cc \longrightarrow \RR$ be a weight function on $\Cc$. The problem of finding the hierarchy $\Hc \subseteq \Cc$ such that $w(\Hc) = \sum_{A \in \Hc} w(A)$ is maximized can be solved in polynomial time.
\end{theorem}
\begin{proof}
If there are any clusters $A \in \Cc$ with negative weight $w(A)$ then set their weights to zero. It follows then that the weight of any maximum hierarchy $\Hc \subseteq \Cc$ equals the weight of the maximum weight hierarchy contained within $\Hc$. The `Hunting for Trees' algorithm of \cite{bry} can now be used to recover the maximum hierarchy of maximum weight.
\end{proof}

\subsection{Parsimony problems on networks}

Consider a set $\Cc$ of clusters on $X$ and let $f: X \rightarrow \Ac$ be a function that assigns each element $x\in X$ a state $f(x)$ in a finite set $\Ac$ ($f$ is referred to in phylogenetics as a  (discrete) {\em character}).  Suppose we have a non-negative function $\delta$ on $\Ac \times \Ac$ where $\delta(a,b)$ assigns a penalty score for changing
state $a$ to $b$ for each pair $a,b \in \Ac$ (the default option is to to take $\delta(a,b)=1$ for all $a\neq b$ and $\delta(a,a)=0$ for all $a$).  

Given any rooted $X-$tree $T$, with vertex set $V$ and arc set $E$,  let $l(f, T, \delta)$ denote the {\em parsimony score} of $f$ on $T$ relative to $\delta$; that is,
$$l(f, T, \delta) = \min_{F: V \rightarrow \Ac, F|X = f}\left\{\sum_{(u,v) \in E} \delta (F(u), F(v))\right\}.$$
In words, $l(f, T, \delta)$ is the minimum sum of $\delta$-penalty scores that are required in order to extend $f$ to an assignment of states to all the vertices of $T$. This quantity can be 
calculated for a given $T$ by well-known dynamic programming techniques (see e.g. \cite{sem}).  
 Let $l(f, \Cc, \delta)$ (respectively, $l_{\rm bin}(f, \Cc)$) denote the minimal value of $l(f,T, \delta)$ among all trees $T$ (respectively, all {\em binary} trees) that have their clusters in $\Cc$.   Then we have the following general result.

\begin{theorem}
\label{app3}
Suppose that  $\Cc$ is contained within a bureaucratic collection $\Cc'$  of subsets of $X$ and  $f: X \rightarrow \Ac$.
There is an algorithm for computing $l(f, \Cc, \delta)$ with running time  polynomial in $n=|X|, |\Ac|$ and $|\Cc'|$.  Moreover, the algorithm can be extended to construct a rooted phylogenetic $X-$tree having all its clusters in $\Cc$ and with parsimony score equal to $l(f, \Cc, \delta)$ in polynomial time.
\end{theorem}

\begin{proof}

For  any subset $Y$ of $X,$ let

$$\delta_Y(a,b) = \begin{cases} \delta(a,b), & \mbox{ if $Y \in \Cc$;}\\0, & \mbox{ if $Y \not \in \Cc$ and $a=b$;}\\ \infty, & \mbox{ otherwise;} \end{cases}$$
and for any rooted phylogenetic $X-$tree $T$, let $$l'(f, T, \delta) := \min_{F: V \rightarrow \Ac, F|X = f}\left\{\sum_{(u,v) \in E} \delta_{c(v)} (F(u), F(v))\right\},$$
where $c(v)$ is the cluster of $T$ associated with $v$.

Let $l'(f, \Cc', \delta)$ (respectively, $l'_{\rm bin}(f, \Cc', \delta)$) be the minimal value of $l'(f, T, \delta)$ over all trees (respectively, all binary trees) with clusters in $\Cc'$.
By the definition of $\delta_Y$, we have:

\begin{equation}
\label{firstidentity}
l(f, \Cc, \delta) = l'(f, \Cc', \delta),
\end{equation}
and by the assumption that $\Cc'$ is bureaucratic we have:
\begin{equation}
\label{seconddentity}
l'(f, \Cc', \delta) = l'_{\rm bin}(f, \Cc', \delta).
\end{equation}
We now describe how $l'_{\rm bin}(f, \Cc', \delta)$ can be efficiently calculated by dynamic programming.

For an element $a \in \Ac$ and $Y \in \Cc'$, let
$L'(Y, a)$ be the minimum value of $l'(f|Y, T, \delta)$ across all binary trees $T$  having leaf set $Y$ and clusters in $\Cc'$, in which the root is assigned state $a$.     

For $|Y|=1$, say $Y = \{y\}$, we have $$L'(Y,a) = \begin{cases}
0, & \mbox{ if } f(y) = a; \\
\infty, & \mbox{ otherwise}
\end{cases}$$
and for $|Y|>1$, we have
\begin{equation}
L'(Y, a) = \min_{Y_1, Y_2 \in \Cc', a_1, a_2 \in \Ac}\left\{L'(Y_1, a_1) + \delta_{Y_1}(a, a_1) + L'(Y_2, a_2) + \delta_{Y_2}(a, a_2):Y_1 \sqcup Y_2 = Y\right\}.
\label{eq:lform}
\end{equation}
Now,  $$l'_{\rm bin} (f, \Cc' \delta) = \min_{a \in \Ac} L'(X,a).$$
Notice that when one evaluates $L'(X,a)$ using the above recursion (Eqn. (\ref{eq:lform})), it  is sufficient to compute $L'(Y,a)$ for just the sets $Y \in \Cc'$ rather than all 
subsets of $X$.   

Thus, in view of Eqns. (\ref{firstidentity}) and (\ref{seconddentity}), one can compute $l(f, \Cc, \delta)$ in time polynomial in $n=|X|, |\Ac|$ and $|\Cc'|$.  Moreover, by suitable book-keeping along the way, one  can
construct a rooted binary phylogenetic $X-$tree with clusters in $\Cc'$ and with a parsimony score equal to $l_{\rm bin}(f, \Cc',  \delta)$;  by collapsing all edges of this tree 
that have a $\delta$-score equal to 0 we obtain a rooted  phylogenetic $X-$tree with clusters in $\Cc$ and with parsimony score equal to $l(f, \Cc,  \delta)$.
\end{proof}

We note that this result has been described in the particular case when $\Cc$ is the bureaucracy described in example (2) above,  and where $f$ maps to a set $A$ with only two elements   \cite{hus2}.   We provide a second application, to phylogenetic networks, based on Example (1) above, of intervals as bureaucratic set systems.  

Let $\Nc$ be a rooted binary phylogenetic network on $X$. We say that $\Nc$ is {\em planar} if it can be drawn in the plane so that all the leaves and the root all lie on the outer face.   Let $sw(\Nc)$ denote the set of `soft-wired' clusters in $\Nc$ (the union of the cluster sets of all trees embedded in $\Nc$; see e.g. \cite{hus}). A simple example is shown in Fig.~\ref{fig1}(c).

\begin{corollary}
 Suppose that $\Nc$ is a  binary and planar phylogenetic network on $X$,  and  $f:X \rightarrow \Ac$. There is an algorithm for computing $l(f, sw(\Nc))$ with running time  polynomial time in $n$.
\end{corollary}
\begin{proof}
Let $x_1,\ldots,x_n$ be the ordering of $X$ given by their positions around the outer face in a planar embedding of $\Nc$, where $x_1$ and $x_n$ come immediately after and before the root. Then any tree $T$ embedded in $\Nc$ can be ordered so that the leaves are in order $x_1,\ldots,x_n$, implying that the clusters of $T$ are all of the form $\{x_i,x_{i+1},\ldots,x_j\}$ for some $1 \leq i \leq j \leq n$. It follows that the set $sw(\Nc)$ is contained in the set of intervals of  $X=\{x_1,\ldots,x_n\}$ (Example 1, above).  The corollary now follows from Theorem~\ref{app3}. 
\end{proof}

We end this paper by posing a computational problem. \\

\noindent {\bf Question.} Is there an algorithm for deciding whether or not $\Cc$ is bureaucratic that runs in time polynomial in $|\Cc|$ and $|X|$?

\bigskip
\noindent{\bf Acknowledgements:} 
We thank the  NZ Marsden Fund and the Allan Wilson Centre for Molecular Ecology and Evolution for supporting this work.

\end{document}